\documentclass{article}
\usepackage{enumitem}
\usepackage{verbatim}
\usepackage{graphicx}
\usepackage{amsfonts,amsmath}
\usepackage{subcaption}
\usepackage{tikz}
\usepackage{amsthm}
\usepackage{authblk,fullpage}
\usepackage[colorlinks]{hyperref}
\usepackage[capitalize]{cleveref}
\definecolor{pinegreen}{rgb}{0.0, 0.47, 0.44}
\hypersetup{citecolor=pinegreen}
\setlist[enumerate]{itemsep=0pt,topsep=1pt}
\setlist[itemize]{itemsep=0pt,topsep=1pt}

\usetikzlibrary{automata,positioning, arrows}
\newcommand{\mayqed}{}
   \usepackage[algo2e, noline, boxed]{algorithm2e}
   \SetKwComment{tcm}{\{}{\}}
    
   \newenvironment{myalgorithm}[2][htbp]
   {%
     \setlength{\algomargin}{.2cm}
     \begin{center}
     \begin{minipage}{#2}
     \begin{algorithm2e}[#1]
     \small
      \let\Par=\par
        \def\par{\endgraf\vspace{.1cm}}
            \SetKw{To}{to}%
        \SetKw{Downto}{downto}%
            \SetKw{Or}{or}%
        \SetKwFor{Algo}{Algorithm}{}{}%
       \vspace{.15cm}%
    }
    {%
      \let\par=\Par\end{algorithm2e}%
      \end{minipage}%
      \end{center}%
    }

  \newcommand{\defproblem}[2]{
  \noindent\begin{center}\fbox{
  \begin{minipage}{0.5\textwidth}
  \vspace{0.1cm}
  #1

  \vspace{0.2cm}
  \noindent
  #2
  \vspace{0.1cm}
  \end{minipage}
  }
  \end{center}
  \vspace{2mm}
}

   \newcommand{\Oh}{\mathcal{O}}
   
   \newcommand{\floor}[1]{\left\lfloor #1 \right\rfloor}

   \newcommand{\integ}{\mathbb{Z}}
   \newcommand{\eps}{\varepsilon}
   \newcommand{\sub}{\subseteq}
   \newcommand{\dB}{dB}
   \newcommand{\occ}{\mbox{\it occ-pos}}

   \renewcommand{\L}{{\mathcal{L}}}
   \newcommand{\CS}{\mbox{\large{\sf{S}}}}
   \newcommand{\Lynd}{\mathit{LynRank}}
   \newcommand{\rot}{\mathit{rot}}
   \newcommand{\minrot}[1]{\left\langle #1 \right\rangle}
   
   \newtheorem{theorem}{Theorem}
   \newtheorem{corollary}[theorem]{Corollary}
   \newtheorem{proposition}[theorem]{Proposition}
   \newtheorem{lemma}[theorem]{Lemma}
   
   \newtheorem{fact}[theorem]{Fact}
   \theoremstyle{definition}   
   \theoremstyle{remark}
   \newtheorem{example}[theorem]{Example}
   \newtheorem*{claim}{Claim}

\author{Tomasz Kociumaka}
\author{Jakub Radoszewski}
\author{Wojciech Rytter}

\affil{Institute of Informatics, University of Warsaw, Poland}
\affil[]{\texttt{[kociumaka,jrad,rytter]@mimuw.edu.pl}}

\date{\vspace{-1cm}}

\title{
Efficient Ranking of Lyndon Words and\\ Decoding
Lexicographically Minimal de Bruijn Sequence\footnote{
This is an extended version of our previous conference paper \cite{DBLP:conf/cpm/KociumakaRR14} with complexities reduced
by a $(\log\sigma)/n$ factor in case of the word RAM model.
}
   }

\begin{document}
\maketitle
\begin{abstract}
We give efficient algorithms for ranking Lyndon words of length $n$ over
an alphabet of size $\sigma$. The rank of a Lyndon word is its position
in the sequence of lexicographically ordered Lyndon words of the same length.
The outputs are integers of exponential size, and
complexity of arithmetic operations on such large integers cannot be ignored.
Our model of computations is the word RAM, in which  basic arithmetic operations on
(large) numbers of size at most $\sigma^n$ take $\Oh(n)$ time.
Our algorithm for ranking Lyndon words makes $O(n^2)$ arithmetic operations
(this would imply directly cubic time on word RAM).
However, using an algebraic approach we
are able to reduce the total time complexity on word RAM to $O(n^2 \log\sigma)$.
We also present an $O(n^3 \log^2 \sigma)$-time algorithm that generates the Lyndon
word of a given length and rank in lexicographic order.
Finally we use the connections between Lyndon words and lexicographically minimal de Bruijn sequences
(a theorem of Fredricksen and Maiorana) to develop the first polynomial-time algorithm
for decoding minimal de Bruijn sequence of any rank $n$
(it determines the position of a given word of length $n$ within the de Bruijn sequence).
\end{abstract}

\section{Introduction}
We consider finite words over an ordered alphabet $\Sigma$ of size $\sigma=|\Sigma|$.
A \emph{Lyndon word}~\cite{Lyndon1954,chen1958free} over $\Sigma$ is a word that is strictly smaller in the lexicographic order than all its
nontrivial cyclic rotations.
For example, for $\Sigma=\{\mathtt{a},\mathtt{b}\}$ where $\mathtt{a}<\mathtt{b}$, the word \texttt{aababb} is a Lyndon word,
as it is smaller than its cyclic rotations: \texttt{ababba}, \texttt{babbaa}, \texttt{abbaab},
\texttt{bbaaba}, \texttt{baabab}.
On the other hand, the word \texttt{abaab} is not a Lyndon word, since its cyclic rotation \texttt{aabab}
is smaller than it.
Also the word \texttt{aabaab} is not a Lyndon word, as its cyclic rotation by 3 letters is equal to it.
Lyndon words have a number of combinatorial properties (see, e.g., \cite{Lothaire}) including the famous Lyndon
factorization theorem~\cite{chen1958free}, which states that every word can be uniquely written
as a concatenation of a lexicographically non-increasing sequence of Lyndon words

(due to this theorem, Lyndon words are also called prime words; see~\cite{Knuth}).
They are also related to \emph{necklaces} of $n$ beads in $k$ colors, that is, equivalence classes of $k$-ary
$n$-tuples under rotation~\cite{FK,fredricksen1978necklaces}.
In particular, a necklace can be identified with the lexicographically minimal tuple in its class, 
and thus it is often defined as a word of length $n$ over alphabet of size $k$ that is smaller than or equal to all its cyclic rotations
(or equivalently, as a power of a Lyndon word of a length that divides $n$).
Lyndon words and necklaces have numerous applications in the field of text algorithms;
see e.g.\ \cite{DBLP:conf/dlt/BonomoMRRS13,ApplicationsOfRuns,TextAlgorithms,DBLP:conf/soda/Mucha13}.

A \emph{de Bruijn sequence of rank $n$}~\cite{deBruijn} is a cyclic sequence of length $\sigma^n$
in which every possible word of length $n$ occurs as a factor exactly once.
For example, for $\Sigma=\{0,1\}$ the following two sequences of length 16 are de Bruijn
sequences of rank 4:
$$0000100110101111 \quad\mbox{and}\quad 0011110110010100.$$
De Bruijn sequences are present in a variety of contexts, such as digital fault testing,
pseudo-random number generation, and modern public-key cryptographic schemes.
There are numerous algorithms for generating such sequences and their generalizations
to other combinatorial structures have been investigated; see~\cite{deBruijnAnalogs,Knuth}.
Fredricksen and Maiorana \cite{DBLP:journals/jct/Fredricksen70,fredricksen1978necklaces} have shown a surprising
deep connection between de Bruijn sequences and Lyndon words:
the lexicographically minimal de Bruijn sequence over $\Sigma$ is a concatenation,
in the lexicographic order, of all Lyndon words over $\Sigma$ whose length is a divisor of $n$.
For example, for $n=6$ and the binary alphabet we have the following decomposition of the minimal
de Bruijn sequence into Lyndon words:
$$0\; 000001\; 000011\; 000101\; 000111\; 001\; 001011\; 001101\; 001111\; 01\; 010111\; 011\; 011111\; 1.$$

\paragraph{\bf Problem definitions and previous results.}
We denote by $\L$ and $\L_n$ the set of all Lyndon words and
all Lyndon words of length $n$, respectively, and define
$$\Lynd(w)\;=\; |\{x \in \L_{|w|}\;:\; x\le w\}|.$$
The problem of ranking Lyndon words can be defined as follows.

\defproblem{\textbf{Problem 1.} Ranking Lyndon words}{Given a Lyndon word $\lambda$, compute $\Lynd(\lambda)$.}

\begin{example}
  For $\Sigma=\{\mathtt{a},\mathtt{b}\}$ we have $\Lynd(\mathtt{ababbb})=8$ since there are 8 Lyndon words of length 6
  that are not greater than $\mathtt{ababbb}$:
  $$\mathtt{aaaaab},\; \mathtt{aaaabb},\; \mathtt{aaabab},\; \mathtt{aaabbb},\; \mathtt{aababb},\; \mathtt{aabbab},\; \mathtt{aabbbb},\; \mathtt{ababbb}.$$
\end{example}

What was previously known is that all Lyndon words of length at most $n$ can be generated in lexicographic order.
The first solution is due to Fredricksen, Kessler, and Maiorana (FKM) \cite{FK,fredricksen1978necklaces};
later Duval developed an alternative algorithm~\cite{DBLP:journals/tcs/Duval88}.
The analysis by Ruskey et al.~\cite{DBLP:journals/jal/RuskeySW92} shows that the FKM
algorithm generates the subsequent Lyndon words in constant amortized time; Berstel and Pocchiola~\cite{DBLP:journals/tcs/BerstelP94} achieved an analogous result for Duval's algorithm.
A different constant-amortized-time solution, based on recursion, was given by Cattell et al.~\cite{DBLP:journals/jal/CattellRSSM00}.
However, there was no polynomial-time algorithm to generate a Lyndon word of an arbitrary rank
or for ranking Lyndon words.
Ruskey stated finding such an algorithm as a research problem in his book~\cite{RuskeyCombGen}.

An intimately related task, of ranking and unranking \emph{necklaces}, was explicitly stated as open by Martínez and Molinero~\cite{Martinez2004}.
As far as obtaining polynomial-time solutions is concerned, one can easily show equivalence of this problem with its counterpart for Lyndon words. 
Listing necklaces  also resembles listing Lyndon words: two of the previously mentioned algorithms (FKM and Cattell et al.'s) 
can be used for both tasks. Variants of the listing problem have also been considered, e.g., generating binary necklaces with a given number of zeroes and ones~\cite{DBLP:journals/siamcomp/RuskeyS99}.

Let $\L^{(n)}=\bigcup_{d \mid n} \L_d$.
By $\dB_n$ we denote the lexicographically first de Bruijn sequence of rank $n$ over the given alphabet $\Sigma$.
It is the concatenation of all Lyndon words in $\L^{(n)}$ in lexicographic order~\cite{DBLP:journals/jct/Fredricksen70,fredricksen1978necklaces}.
For a word $w$ of length $n$ over $\Sigma$, by $\occ(w,\dB_n)$ we denote the (1-based) position of its occurrence in $\dB_n$.
The problem of decoding the minimal de Bruijn sequence can be stated as follows.

\defproblem{\textbf{Problem 2.} Decoding minimal de Bruijn sequence}{Given a word $w$ over $\Sigma^n$, compute $\occ(w,\dB_n)$.}

\begin{example}
  For $\Sigma=\{0,1\}$ we have $\dB_4=0000100110101111$.
  For this sequence:
  $$\occ(1001,\dB_4)=5,\quad \occ(0101,\dB_4)=10,\quad \occ(1100,\dB_4)=15.$$
\end{example}

For several types of de Bruijn sequences, there exist polynomial-time decoding algorithms \cite{DBLP:journals/tit/MitchellEP96,DBLP:journals/dm/Tuliani01}.
They find the position of an arbitrary word of length $n$ in a specific de Bruijn sequence,
which proves useful in certain types of position sensing applications of de Bruijn sequences (see~\cite{DBLP:journals/dm/Tuliani01}).
Nevertheless, no polynomial-time decoding algorithm for lexicographically minimal de Bruijn sequence was known prior to our contribution.
Note that the FKM algorithm can be used to compute the subsequent symbols of the lexicographically minimal
de Bruijn sequence with worst-case $\Oh(n^2)$ delay~\cite{FK} and amortized $\Oh(1)$ delay~\cite{DBLP:journals/jal/RuskeySW92}.
Alternative solutions achieve $\Oh(n)$~\cite{DBLP:journals/jct/Fredricksen72,DBLP:journals/jal/Ralston81} or even $\Oh(1)$~\cite{JRmgr} worst-case delay.
All these solutions only allow to generate characters of $\dB_n$ \emph{in order}, though.

\paragraph{\bf Our model of computations.}
Our algorithms work in the \emph{word RAM} model; see \cite{DBLP:conf/stacs/Hagerup98}.
In this model, we assume that $\sigma$ and $n$ fit in a single machine word;
in other words, a single machine word has at least $\max(\log \sigma,\log n)$ bits and simple arithmetic operations on
small numbers (i.e., numbers which fit in a constant number of machine words) are performed in constant time.
Basic arithmetic operations (addition, subtraction, multiplication by a small number) on numbers of size at most $\sigma^n$ take $\Oh(\frac{\log \sigma^n}{\max(\log n, \log \sigma)}) = \Oh(n)$ time.

Another model of computation is the \emph{unit-cost RAM}, where each arithmetic operation
takes constant time. This model is rather unrealistic if we deal with large numbers.
However, it is a useful intermediate abstraction.

\paragraph{\bf Our results.}
We present an $\Oh(n^2 \log \sigma)$-time solution for finding the rank of a Lyndon word
(Problem~1).
The algorithm actually computes $\Lynd(w)$ for arbitrary $w$ that are not necessarily Lyndon words.
Using binary search, it yields an $\Oh(n^3 \log^2 \sigma)$-time algorithm for computing
the $k$-th Lyndon word of length $n$ (in the lexicographic order) for a given $k$.
Next, we show an $\Oh(n^2\log \sigma)$-time solution for decoding minimal de Bruijn sequence $\dB_n$ (Problem~2).
We also develop an $\Oh(n^3 \log^2 \sigma)$-time algorithm computing the $k$-th symbol of
$\dB_n$ for a given $k$.
Additionally, we obtain analogous results for a variant $\dB'_n$ of the minimal de Bruijn sequence, introduced by Au~\cite{DBLP:journals/dm/Au15},
in which all factors of length $n$ are primitive and every primitive word of length $n$ occurs exactly once.
All these algorithms work in the word RAM model.
In the unit-cost RAM, the time complexities reduce by a factor of $\log \sigma$.

\paragraph{\bf Related work.}
A preliminary version of this paper
appeared as \cite{DBLP:conf/cpm/KociumakaRR14}.
At about the same time, similar results were published by Kopparty, Kumar, and Saks \cite{DBLP:conf/icalp/KoppartyKS14}.
The work in these two papers was done independently.
The papers provide polynomial-time algorithms for ranking Lyndon words and necklaces, respectively,
and these two problems can be easily reduced  to each other.
The authors in \cite{DBLP:conf/icalp/KoppartyKS14}
put the results in a broader context and have some additional applications
(indexing irreducible polynomials and explicit constructions of certain algebraic objects).
On the other hand, we exercised more care in designing the algorithm to obtain a better polynomial running time.
In particular, \cite{DBLP:conf/cpm/KociumakaRR14} contained an $\Oh(n^3)$-time
algorithm for ranking Lyndon words in the word RAM model,
which works in $\Oh(n^2)$ time in the unit-cost RAM.
We also obtained a cleaner approach to alphabets of size more than 2.
An alternative $\Oh(n^2)$-time algorithm in the unit-cost RAM model was later designed by Sawada and Williams~\cite{DBLP:journals/jda/SawadaW17}.

\paragraph{\bf Structure of the paper.}
\Cref{sec:prelim,sec:comb,sec:auto,sec:arith} (and \ref{sec:debruijn}) contain a full version of the paper \cite{DBLP:conf/cpm/KociumakaRR14}.
\Cref{sec:prelim} defines the notions of self-minimal words (necklaces) and Lyndon words and lists a number of their properties.
In \cref{sec:comb} we use combinatorial tools to obtain a formula for $\Lynd(w)$ in the case that $w$ is self-minimal.
The next three sections are devoted to efficient computation of the main ingredient of this formula.
In \cref{sec:auto} we show that it is sufficient to count specific walks in an auxiliary automaton.
Then in \cref{sec:arith,sec:wordram} we show efficient implementations of this technique under unit-cost RAM
and word RAM models, respectively.
In \cref{sec:debruijn}, we apply ranking of Lyndon words to obtain efficient decoding of minimal de Bruijn sequence.

\section{Preliminaries}\label{sec:prelim}
Let $\Sigma$ be an ordered alphabet of size $\sigma=|\Sigma|$.
By $\Sigma^*$ and $\Sigma^n$, we denote the set of all finite words over $\Sigma$ and the set of all such words of length~$n$.
The empty word is denoted as $\varepsilon$.
If $w$ is a word, then $|w|$ denotes its length, $w[i]$ its $i$-th letter
(for $1 \le i \le |w|$), $w[i,j]$ its factor $w[i] w[i+1] \cdots w[j]$
and $w_{(i)}$ its prefix $w[1,i]$.
A suffix of $w$ is a word of the form $w[i,n]$.
A prefix or a suffix is called proper if it is shorter than $w$.
By $w^k$ we denote a concatenation of $k$ copies of $w$.
Any two words can be compared in the lexicographic order:
$u$ is smaller than $v$ if $u$ is a proper prefix of $v$ or if the letter following the longest common prefix
of $u$ and $v$ in $u$ is smaller than in $v$.

By $\rot(w,c)$ let us denote a \emph{cyclic rotation} of $w$
obtained by moving $c \bmod{|w|}$ first letters of $w$ to its end (preserving the order of the letters).
We say that the words $w$ and $\rot(w,c)$ are \emph{cyclically equivalent}
(sometimes called \emph{conjugates}).
By $\minrot{w}$ we denote the lexicographically minimal cyclic rotation of $w$.
A word $w$ is called \emph{self-minimal} (alternatively, a \emph{necklace}) if $\minrot{w}=w$.
The following fact gives a simple property of self-minimal words.

\begin{fact}\label{fct:lynd}
  If $w \in \Sigma^n$ is self-minimal and $d \mid n$, then
  $(w_{(d)})^{n/d} \le w$.
\end{fact}
\begin{proof}
  Assume to the contrary that $(w_{(d)})^{n/d} > w$.
  Let $k$ be the index of the first letter where these two words differ.
  Then of course $(w_{(d)})^{n/d}[k] > w[k]$.
  Let $j$ be an integer defined as $jd+1 \le k \le (j+1)d$.
  Then $w_{(d)} > w[jd+1,(j+1)d]$.
  Hence, $w>\rot(w,jd)$, a contradiction.
\mayqed\end{proof}

%

In the ranking algorithms that we design below, we make an assumption that
the input word is self-minimal.
Consequently, we often need to replace a given arbitrary word $w$ with
the lexicographically largest self-minimal word $w'$ (of the same length) not exceeding $w$.
In the construction of this routine, we use the following auxiliary facts.

\begin{fact}[see \cite{DBLP:journals/jal/Duval83}]\label{fct:minrot}
  For a given word $x\in \Sigma^n$, the lexicographically minimal cyclic rotation of $x$ ($\minrot{x}$) and the lexicographically minimal suffix of $x$ can be computed in $\Oh(n)$ time.
\end{fact}

  \begin{fact}\label{fct:prvaux}
    Let $x$ and $x'$, $x < x'$, be words of length $n$ with the longest common prefix of length $p$.
    If $x$ is self-minimal and $x'[i]=\max\Sigma$ for each $i > p+1$, then $x'$ is self-minimal.
  \end{fact}
  \begin{proof}
    First, note that $x[1]<z$ where $z=\max\Sigma$.
    Indeed, if $x[1]=z$, then we would have $x=z^{n}$ by self-minimality of $x$.
    However, this contradicts $x'>x$.
    Now, observe that $x'[1]=z$ implies $p=0$ and consequently $x'[i]=z$ for $i>1$.
    Thus, $x'=z^n$ is trivially self-minimal.
    Hence, from now on we may assume that $x'[1]<z$.

    Since $x'[i]=z$ for $i>p+1$, this means that $x'<z<\rot(x',i)$ for $i\in \{p+1,\ldots,n-1\}$.
    Thus, it suffices to show that $x'\le \rot(x',i)$ for $i\in \{1,\ldots,p\}$.
    For a proof by contradiction, suppose that $x' > \rot(x',i)$. Consequently,
    \[x' > \rot(x',i) > \rot(x,i) \ge x\]
    since for $i\in\{1,\ldots,p\}$ we have $\rot(x,i)<\rot(x',i)$ with the longest common prefix of length exactly $p-i$.
    However, the obtained sequence of inequalities proves that $\rot(x',i)$ and $\rot(x,i)$
    have a common prefix of length at least $p$ due to such a common prefix of $x$ and $x'$.
    The contradiction concludes the proof.
  \mayqed\end{proof}

We are now ready to implement the announced procedure.

\begin{lemma}\label{lem:prv}
  For a given word $w\in \Sigma^n$ we can compute in $\Oh(n^2)$ time the
  lexicographically largest self-minimal word $w'\in \Sigma^n$ such that $w'\le w$.
\end{lemma}
\begin{proof}
  \Cref{fct:minrot} lets us check whether $w$ is self-minimal. If so, we simply return $w'=w$.
  Consequently, we may assume that the sought word $w'$ is strictly smaller than $w$.
  Assume the longest common prefix of $w$ and $w'$ 
  is $w_{(k-1)}$ for some $k\le n$.
  Then $b=w[k] > w'[k]$, so
  in particular $w[k]\ne \min \Sigma$ and one can choose $b'\in\Sigma$
  as the letter preceding $b$. Additionally, let $z=\max\Sigma$.
  Consider a word $$w''=w_{(k-1)} b' z^{n-k}.$$
  Note that $w' \le w'' < w$.
  If $w'<w''$, then \cref{fct:prvaux} applied for $x=w'$ and $x'=w''$
  would show that $\minrot{w''}=w''$, and this would contradict the definition of $w'$.
  Hence, $w'=w''$.
    
  Consequently, it suffices to consider $w$ and,
  for each $k\in\{1,\ldots,n\}$ such that $w[k]\ne \min \Sigma$,
  a word $w_{(k-1)}b'z^{n-k}$ where $b'$ is the letter preceding $w[k]$ in $\Sigma$.
  Since $w'$ is guaranteed to be one of the considered words,
  it suffices to output the largest of these candidates for which $\minrot{w'}=w'$.
  This procedure can be implemented in $\Oh(n^2)$ time using \cref{fct:minrot}.
\mayqed\end{proof}

The shortest word $u$ such that $w=u^k$ for some positive integer $k$ is called the
\emph{primitive root} of $w$. We say that $w$ is \emph{primitive} if its primitive root is $w$ itself.
Otherwise, $w$ is called non-primitive.
The primitive root of a word of length $n$ can be computed in $\Oh(n)$ time 
using the failure function from the algorithm of Knuth, Morris, and Pratt~\cite{DBLP:journals/siamcomp/KnuthMP77};
see~\cite{AlgorithmsOnStrings}.

We say that $\lambda \in \Sigma^*$ is a \emph{Lyndon word} if it is primitive and self-minimal.
Equivalent definitions are that a Lyndon word is (strictly) smaller than all its suffixes or all its cyclic rotations.
All cyclic rotations of a Lyndon word are different primitive words.
Moreover, every self-minimal word can be expressed in a unique way as $\lambda^k$
for some Lyndon word $\lambda$; see \cite{Lothaire}.
Below we show an additional property of Lyndon words that will be useful in \cref{sec:debruijn}.

\begin{fact}\label{fct:lyndon}
  Let $\lambda_1,\lambda_2 \in \L^{(n)}$.
  \begin{enumerate}[label={(\alph*)}]
    \item\label{aa}
      It is not possible that $\lambda_1 < \lambda_2 \le \lambda_1^{n/|\lambda_1|}$.
    \item\label{bb}
      If $\lambda_1 < \lambda_2$, then $\lambda_1^{n/|\lambda_1|} < \lambda_2^{n/|\lambda_2|}$.
  \end{enumerate}
\end{fact}
\begin{proof}
  \ref{aa} The inequalities imply that $\lambda_1$ is a proper prefix of $\lambda_2$.
  Let $\lambda_2 = \lambda_1^k x$, where $k \ge 1$ is an integer and $\lambda_1$ is not a prefix of $x$.
  We have
  $$\lambda_2 \le \lambda_1^{n/|\lambda_1|}\ \Longrightarrow\ x \le \lambda_1^{n/|\lambda_1|-k}.$$
  If $|x| < |\lambda_1|$, then we conclude that $x < \lambda_1$.
  Otherwise, $x=x'x''$, where $|x'| = |\lambda_1|$ and $x' \ne \lambda_1$.
  Hence, $x' < \lambda_1$, so again $x < \lambda_1$.
  In both cases we have $x < \lambda_1 < \lambda_2$, which contradicts
  the fact that a Lyndon word is smaller than all its suffixes.

  \ref{bb} Suppose to the contrary that $\lambda_1 < \lambda_2$ but $\lambda_1^{n/|\lambda_1|} \ge \lambda_2^{n/|\lambda_2|}$.
  Then
  $$\lambda_1 < \lambda_2 \le \lambda_2^{n/|\lambda_2|} \le \lambda_1^{n/|\lambda_1|}.$$
  This contradicts part \ref{aa}.
\mayqed\end{proof}

\section{Combinatorics of Ranking Lyndon Words}\label{sec:comb}
Recall that, for a word $w\in \Sigma^n$, we defined $\Lynd(w)$ as the number of Lyndon words
in $\Sigma^n$ not exceeding $w$.
Our basic goal (stated as Problem~1) is to efficiently compute $\Lynd(w)$ for a given word $w$.
It suffices to compute $\Lynd(w)$ for a self-minimal word $w$.
If $w$ is not self-minimal, then $\Lynd(w)=\Lynd(w')$ where $w'$ is the greatest self-minimal
word such that $w' \le w$; such $w'$ can be computed efficiently using \cref{lem:prv}.

We will show how to reduce computation of $\Lynd(w)$ to the computation 
of the cardinality of the following set:
$$\CS(v) = \{x \in \Sigma^{|v|} : \minrot{x}\le v\}$$
for some prefixes $v$ of $w$.

\begin{example}
  For $\Sigma=\{\mathtt{a},\mathtt{b}\}$,
  there are seven words of length four lexicographically smaller than or equal to $w=\mathtt{abba}$:
  $$\mathtt{aaaa},\mathtt{aaab},\mathtt{aaba},\mathtt{aabb},\mathtt{abaa},\mathtt{abab},\mathtt{abba}.$$
  This set contains words from the following four equivalence classes.
  Each class includes a self-minimal word that is underlined.
    $$\{\mathtt{\underline{aaaa}}\}\cup\{\mathtt{\underline{aaab}},\mathtt{aaba},\mathtt{abaa}\}\cup\{\mathtt{\underline{aabb}},\mathtt{abba}\}\cup\{\mathtt{\underline{abab}}\}.$$
  Thus, $\CS(w)$ consists of four full classes of cyclic equivalence:
  $$\CS(w)=\{\mathtt{\underline{aaaa}}\}\cup
             \{\mathtt{\underline{aaab}},\mathtt{aaba},\mathtt{abaa},\mathtt{baaa}\}\cup
             \{\mathtt{\underline{aabb}},\mathtt{abba},\mathtt{bbaa},\mathtt{baab}\}\cup
             \{\mathtt{\underline{abab}},\mathtt{baba}\}.$$
\end{example}

Let us introduce the following auxiliary sets defined for $w\in \Sigma^n$ and divisors $d \mid n$:
\begin{align*}
  \CS_d(w)  &= \{x \in \Sigma^d : \minrot{x}^{n/d}\le w \}\\
  \CS'_d(w) &= \{x \in \Sigma^d : x\mbox{ is primitive},\,\minrot{x}^{n/d}\le w\}.
\end{align*}

\begin{example}
  For $w=\mathtt{abbaaa}$ and $\Sigma=\{\mathtt{a},\mathtt{b}\}$, we have
  \begin{align*}
  \CS_2(w)&=\{\mathtt{aa,ab,ba}\} & \CS'_2(w)&=\{\mathtt{ab},\mathtt{ba}\}\\
  \CS_3(w)&=\{\mathtt{aaa},\mathtt{aab},\mathtt{aba},\mathtt{baa}\} & \CS'_3(w)&=\{\mathtt{aab},\mathtt{aba},\mathtt{baa}\}
  \end{align*}
\end{example}

As shown in the following two facts, $\Lynd(w)$ is closely related to $|\CS'_n(w)|$,
which can be expressed in terms of $|\CS_d(w)|$ for $d \mid n$.

\begin{fact}\label{fct:lyndtocsprime}
For every word $w\in \Sigma^n$, we have $\Lynd(w) = \frac{1}{n}|\CS'_{n}(w)|$.
 \end{fact}
  \begin{proof}
    Observe that $\CS'_n(w)$ is the set of all primitive words of length $n$ that have a cyclic
    rotation not exceeding $w$. Each Lyndon word of length $n$ not exceeding $w$ 
    corresponds to $n$ such words: all its cyclic rotations.
  \mayqed\end{proof}

\begin{fact}\label{fct:csprimetocs}
For every word $w\in \Sigma^n$, if $d \mid n$, then $|\CS'_d(w)| = \sum_{\ell \mid d} \mu(\tfrac{d}{\ell}) |\CS_\ell(w)|$.
\end{fact}
\begin{proof}
    We first show that
    \begin{equation}\label{eq:CS_CS1}
      |\CS_\ell(w)| = \sum_{d \mid \ell} |\CS'_d(w)|.
    \end{equation}
    For a word $x$ of length $\ell$ there exists exactly one primitive word $y$ such that $y^k=x$ where $k\in \integ_+$.
    Thus:
    $$\CS_\ell(w) = \bigcup_{d \mid \ell} \left\{y \in \Sigma^d : y\mbox{ is primitive}, \minrot{y^{\ell/d}}^{n/\ell}\le w \right\},$$
    and the sum is disjoint.
    Now $\minrot{y^{\ell/d}}^{n/\ell} = \minrot{y}^{n/d}$ implies \eqref{eq:CS_CS1}.
    From this formula, we obtain the claimed equality by the Möbius inversion formula.
\end{proof}

The $\CS_d(w)$ sets are closely related to the regular $\CS(v)$ sets for prefixes $v$ of $w$.
It is most evident for a self-minimal word.

\begin{fact}\label{fct:C1_C}
  If $w \in \Sigma^n$ is self-minimal and $d \mid n$, then $\CS_d(w) = \CS(w_{(d)})$.
\end{fact}
\begin{proof}
  If $d=n$, the equality of the two sets is trivial. Assume $d<n$.
  Let us prove the equality by showing both inclusions.

  Assume that $x \in \CS_d(w)$.
  This means that $\minrot{x}^{n/d} \le w$, therefore $\minrot{x} \le w_{(d)}$ (as $|x|=d$).
  Hence, $x \in \CS(w_{(d)})$.

  Now assume that $x \in \CS(w_{(d)})$.
  This means that $\minrot{x} \le w_{(d)}$.
  We have
  $\minrot{x}^{n/d} \le (w_{(d)})^{n/d} \le w$
  where the second inequality is due to \cref{fct:lynd}.
  Hence, $x \in \CS_d(w)$.
\mayqed\end{proof}

The facts that we have just proved let us derive a formula for $\Lynd(w)$.

\begin{lemma}\label{lem:formula}
  If a word $w\in \Sigma^n$ is self-minimal, then
  $$\Lynd(w) = \tfrac{1}{n}\sum_{d\mid n}\mu(\tfrac{n}{d})\left|\CS(w_{(d)})\right|.$$
\end{lemma}
\begin{proof}
We use Facts \ref{fct:lyndtocsprime}, \ref{fct:csprimetocs}, and \ref{fct:C1_C} in a series of transformations to express
$\Lynd(w)$ using $|\CS'_n(w)|$, $|\CS_d(w)|$ for $d\mid n$, and finally $|\CS(w_{(d)})|$ for $d \mid n$:
$$\Lynd(w)=\tfrac1n \left|\CS'_n(w)\right| = \tfrac1n \sum_{d \mid n} \mu(\tfrac{n}{d})\left|\CS_d(w)\right| = \tfrac{1}{n}\sum_{d\mid n}\mu(\tfrac{n}{d})\left|\CS(w_{(d)})\right|.$$
\mayqed\end{proof}

\begin{example}
  Let $w\,=\,\texttt{ababbb}$. We have  $w_{(1)}=\texttt{a},\; w_{(2)}=\texttt{ab},\; w_{(3)}=\texttt{aba}$ and 
  \begin{align*}
    \CS(w_{(1)})&=\{\texttt{a}\},&\quad \CS(w_{(2)})&=\{\texttt{aa},\texttt{ab},\texttt{ba}\},\\
    \CS(w_{(3)})&=\{\texttt{aaa},\texttt{aab},\texttt{aba},\texttt{baa}\},&\quad |\CS(w)|&=54,
  \end{align*}
  \begin{multline*}
    \Lynd(w)=  \tfrac{1}{6}\cdot \bigl(\mu(1) \left|\CS(w)\right|+\mu(2)\left|\CS(w_{(3)})\right|+ \mu(3)\left|\CS(w_{(2)})\right|\\+\mu(6)\left|\CS(w_{(1)})\right|\bigr)
    =  \tfrac{1}{6}\cdot (54-4-3+1)=8.
  \end{multline*}
  The set of Lyndon words of length $6$ that are not greater than $w=\texttt{ababbb}$ is:
  $$\{\mathtt{aaaaab},\; \mathtt{aaaabb},\; \mathtt{aaabab},\; \mathtt{aaabbb},\; \mathtt{aababb},\; \mathtt{aabbab},\; \mathtt{aabbbb},\; \mathtt{ababbb}\}.$$
  Indeed, it contains eight elements.
\end{example}

The next three sections are devoted to a proof of the following lemma.

\begin{lemma}\label{lem:csgen}
  For a self-minimal word $w\in \Sigma^n$, one can compute $|\CS(w)|$:
  \begin{enumerate}[label={(\alph*)}]
    \item\label{aaa} in $\Oh(n^2)$ time in the unit-cost RAM,
    \item\label{bbb} in $\Oh(n^2 \log \sigma)$ time in the word RAM.
  \end{enumerate}
\end{lemma}

\noindent
As a consequence of this lemma, we obtain efficient ranking of Lyndon words.

\begin{fact}\label{fct:div}
  If $\alpha>1$ is a real constant,
  then $\sum_{d \mid n} d^\alpha = \Oh(n^\alpha)$.
\end{fact}
\begin{proof}
  Recall that for $\alpha>1$ we have $\sum_{n=1}^\infty \frac{1}{n^\alpha} = \Oh(1)$. Consequently,
  $$\sum_{d\mid n} d^\alpha = \sum_{d\mid n}\left(\tfrac{n}{d}\right)^\alpha \le \sum_{d=1}^\infty \left(\tfrac{n}{d}\right)^\alpha = n^\alpha \sum_{d=1}^\infty \tfrac{1}{d^\alpha} = \Oh(n^\alpha).\quad\quad$$

\vspace*{-1.1cm}\mayqed\end{proof}

\begin{theorem}\label{thm:lynd}
  For an arbitrary word $w$ of length $n$, one can compute $\Lynd(w)$  in $\Oh(n^2\log\sigma)$ time in the word RAM or
  in $\Oh(n^2)$ time in the unit-cost RAM.
\end{theorem}
\begin{proof}
  We use the formula given by \cref{lem:formula} and the algorithm of \cref{lem:csgen}.
  If any of the words $w$, $w_{(d)}$ is not self-minimal, then instead we take the greatest word of the same length that is
  not greater than it and is self-minimal (using \cref{lem:prv}).
  The time complexity is $\Oh(\sum_{d \mid n}d^2\log \sigma)$ in the word RAM
  or $\Oh(\sum_{d \mid n}d^2)$ in the unit-cost RAM which, by \cref{fct:div}, reduces to $\Oh(n^2 \log \sigma)$ in the word RAM
  or $\Oh(n^2)$ in the unit-cost RAM, respectively.
\mayqed\end{proof}

We also obtain an efficient algorithm for ``unranking'' Lyndon words.

\begin{theorem}\label{thm:kth}
  The $k$-th Lyndon word of length $n$ can be found in $\Oh(n^3\log^2\sigma)$ time in the word RAM
  or $\Oh(n^3 \log\sigma)$ time in the unit-cost RAM.
\end{theorem}
\begin{proof}
  By definition of the $\Lynd$ function, we are looking for the smallest $w\in \Sigma^n$  such that $\Lynd(w)\ge k$.
  We binary search $\Sigma^n$ with respect to the lexicographic order, using the algorithm of \cref{thm:lynd}
  to check whether $\Lynd(w)\ge k$.
  The size of the search space is $\sigma^n$,  which gives an additional $n\log\sigma$-time factor.
\mayqed\end{proof}

\section{Automata-Theoretic Interpretation}\label{sec:auto}
From now on we assume that $w$ is self-minimal.
Our goal is to compute $|\CS(w)|$.

\newcommand{\Pref}{\mathrm{Pref}}
Let $\Pref_{-}(w)=\{w_{(i)}s : i\in\{0,\ldots,n-1\},s\in \Sigma, s<w[i+1]\}\cup \{w\}$.
Consider a language $L(w)$ containing words that have a factor $y\in \Pref_{-}(w)$.
Equivalently, $x\in L(w)$ if there exists a factor of $x$ which is smaller than or equal to $w$,
but is not a proper prefix of $w$.
For a language $L\sub \Sigma^*$, let $\sqrt{L} = \{x : x^2 \in L\}$.
\begin{fact}\label{fct:char}
  $\CS(w)=\sqrt{L(w)}\cap \Sigma^{n}$.
\end{fact}
\begin{proof}
  Consider a word $x \in \Sigma^{n}$. 
  If $x \in \CS(w)$, then $\minrot{x}\le w$.
  Take $y=\minrot{x}$, which is a factor of $x^2$.
  Some prefix of $y$ belongs to $\Pref_{-}(w)$.
  This prefix is a factor of $x^2$, so $x^2\in L(w)$.
  Consequently, $x \in \sqrt{L(w)}$.

  On the other hand, assume that $x \in \sqrt{L(w)}$, so $x^2$ contains a factor $y\in \Pref_{-}(w)$.
  Let us fix the first occurrence of $y$ in $x^2$. Observe that $y$ can be extended to a cyclic rotation $x'$ of $x$.
  Note that $y\in \Pref_{-}(w)$ implies that $x'\le w$,
  hence $\minrot{x}\le x'\le w$ and $x \in \CS(w)$.
\mayqed\end{proof}

\noindent
We construct a deterministic finite automaton $A=(Q,\Sigma,\delta,q_0,F)$ recognizing $L(w)$.
It has $|Q|=n+1$ states: one for each prefix of $w$.
The initial state is $q_0 = w_{(0)}$ and the only accepting state (the only element of the set $F$) is $w_{(n)}=AC$.
The transitions are defined as follows: we set $\delta(AC,c)=AC$ for any $c \in \Sigma$ and 
$$
\delta(w_{(i)}, c) = \begin{cases}
w_{(0)}&\text{if }c>w[i+1],\\
w_{(i+1)}&\text{if }c=w[i+1]\text{ and }i\ne n-1,\\
AC&\text{otherwise}.
\end{cases}
$$
\Cref{fig:automaton} contains an example of such an automaton.

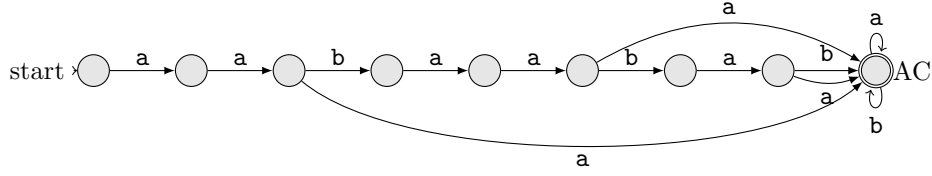
\begin{figure}[htbp]
\begin{center}
\begin{tikzpicture}[scale=.1,-latex,node distance=1.3,on grid,auto] 

  \tikzstyle{every state}=[fill=black!10,draw=black, inner sep = 0, inner sep=0, minimum size = 12]

\node[initial,state] (w0) {};
\node[state] (w1) [right =of w0] {};
\node[state] (w2) [right = of w1] {};
\node[state] (w3) [right = of w2] {};
\node[state] (w4) [right = of w3] {};
\node[state] (w5) [right = of w4] {};
\node[state] (w6) [right = of w5] {};
\node[state] (w7) [right = of w6] {};
\node[state, accepting] (AC) [right = of w7] {};
\draw (AC) node[right=0.1cm] {AC};

\path (w0) edge node[above=-.05] {\texttt{a}} (w1);
\path (w1) edge node[above=-.05] {\texttt{a}} (w2);
\path (w2) edge node[above=-.05] {\texttt{b}} (w3);
\path (w3) edge node[above=-.05] {\texttt{a}} (w4);
\path (w4) edge node[above=-.05] {\texttt{a}} (w5);
\path (w5) edge node[above=-.05] {\texttt{b}} (w6);
\path (w6) edge node[above=-.05] {\texttt{a}} (w7);
\path (w7) edge node {\texttt{b}} (AC);
\path (w7) edge [bend right=20] node[below] {\texttt{a}} (AC);
\path (AC) edge [loop above] node {\texttt{a}} (AC);
\path (AC) edge [loop below] node {\texttt{b}} (AC);
\path[draw =black] (w2) .. controls ([{shift=(-40:20)}]w2) and ([{shift=(220:20)}]AC) .. node[below] {\texttt{a}} (AC);
\path (w5) edge [bend left] node {\texttt{a}} (AC);
 
\end{tikzpicture}
\caption{\label{fig:automaton}
   Automaton $A$ that accepts $L(w)$
   for a word $w=\mathtt{aabaabab}$ and alphabet $\Sigma=\{\mathtt{a},\mathtt{b}\}$.
   Missing links lead to the initial state.
}
\end{center}
\end{figure}
 
Note that all accepting paths in the automaton have a simple structure.
Each of them can be divided into fragments, each of which is a path that starts in $w_{(0)}$,
visits a number of states corresponding to subsequent prefixes of $w$ and eventually goes either back to $w_{(0)}$ or to $AC$.
In the latter case the word spelled by the path fragment is an element of $\Pref_{-}(w)$.
After the path reaches $AC$, it stays there.
Hence, if a word $x$ is accepted by the automaton, then it contains a factor from $\Pref_{-}(w)$, so $x \in L(w)$.
Consequently, $L(A) \subseteq L(w)$.
By a more thorough analysis we show below that $L(A)=L(w)$.
\begin{lemma}
  Let $x\in \Sigma^*$ and let $q$ be the state of $A$ after reading $x$.
  If $x\in L(w)$, then $q=AC$. Otherwise, $q$ corresponds to the longest prefix of $w$
  which is a suffix of $x$.
\end{lemma}
\begin{proof}
  The proof goes by induction on $|x|$.
  If $|x|=0$, the statement is clear.
  Consider a word $x$ of length $|x| \ge 1$. Let $x=x'c$ where $c\in \Sigma$.
  If $x'\in L(w)$, then clearly $x\in L(w)$. By inductive assumption after reading $x'$ the automaton
  is in $AC$, and $A$ is constructed so that it stays in $AC$ once it gets there.
  Thus, the conclusion holds in this case.
  From now on we assume that $x'\notin L(w)$.

  Let $w_{(i)}$ be the state of $A$ after reading $x'$. If $c < w[i+1]$, clearly $x\in L(w)$
  ($y=w_{(i)}c \in \Pref_{-}(w)$), and the automaton proceeds to $AC$ as desired.
  Similarly, it behaves correctly if $i=n-1$ and $c=w[i+1]$.
  Consequently, we may assume that $c\ge w[i+1]$ and that $w$ is not a suffix of $x$.

  Take any $j$ such that $w_{(j)}$ is a suffix of $x'$ (possibly empty).
  Note that then $w_{(j)}$ is a suffix of $w_{(i)}$. 
  Consequently, $w_{(j)}w[i+1,n]w_{(i-j)}$ is a cyclic rotation of $w$,
  so
  $$w_{(j)}w[i+1,n]w_{(i-j)}\ge \minrot{w}=w=w_{(j)}w[j+1,n].$$
  Hence, $c \ge w[i+1]\ge w[j+1]$.
  This implies that $w_{(j)}c$
  could be a prefix of $w$ only if $c=w[i+1]=w[j+1]$.
  In particular, $A$ indeed shifts to the longest prefix of $w$ being a suffix of $x$.
  Now we only need to prove that $x\notin L(w)$.
  For a proof by contradiction, choose a factor $y$ of $x$ such that $y\in \Pref_{-}(w)$ and $|y|$ is minimal.
  Note that $y$ is a suffix of $x$ (since $x'\notin L(w)$). We have $y=w_{(j)}c$ for some $j\le n-1$
  and $c<w[j+1]$. As we have already noticed, such a word cannot be a suffix of $x$.
\mayqed\end{proof}

We say that an automaton with the set of states $Q$ is \emph{sparse}
if the underlying directed graph has $\Oh(|Q|)$ edges counting parallel edges as one.
Note that the transitions from any state $q$ of $A$ lead to at most 3 different
states, so $A$ is sparse.

The following corollary summarizes the construction of $A$.
\begin{corollary}\label{cor:auto}
  Let $w\in \Sigma^n$ be a self-minimal word. One can construct a sparse 
  automaton $A$ with $\Oh(n)$ states recognizing $L(w)$.
\end{corollary}

Let us use the natural extension of the transition function of an automaton into words:
$$\delta(q,x) = \delta(\ldots\delta(\delta(q,x[1]),x[2])\ldots,x[k]) \quad\mbox{for }x \in \Sigma^k.$$
For states $q,q'\in Q$ let us define the set $L_A(q,q')=\{x\in \Sigma^* : \delta(q,x)=q'\}$
of the labels of walks from $q$ to $q'$.
The following lemma shows a crucial property of the words $x^2$ from the language $L(A)$
such that $x \not\in L(A)$.
It makes use of the special structure of the automaton $A$.

\begin{lemma}\label{lem:comb_crucial}
  Let $x\in \Sigma^n$. If $x^2\in L(A)$ but $x\notin L(A)$,
  then there is a unique decomposition $x=x_1x_2x_3$ such that $x_1,x_3 \ne \eps$, $x_3x_1\in \Pref_{-}(w)$ and
  $x_1x_2 \in L_A(w_{(0)},w_{(0)})$. 
\end{lemma}
\begin{proof}
  Let $va$ (for $v\in \Sigma^*, a\in \Sigma$) be the shortest prefix of $x^2$ which belongs to $L(A)$.
  Let $w_{(k)}=\delta(w_{(0)},v)$ be the state of $A$ after reading $v$. 
  Also, let $u$ be the prefix of $v$ of length $|v|-k$. 
  The structure of the automaton implies that $\delta(w_{(0)},u)=w_{(0)}$ and that
  $u$ is actually the longest prefix of $x^2$ which belongs to $L_A(w_{(0)},w_{(0)})$.
  Note that $v=uw_{(k)}$ and $w_{(k)}a\in \Pref_{-}(w)$, so $x\notin L(A)$ implies $|u|<n \le |v|$.
  We set the decomposition so that $x_1x_2=u$ and $x_3x_1=w_{(k)}a$.
  Uniqueness follows from deterministic behaviour of the automaton.
\mayqed\end{proof}

\begin{example}
  Let $w=\mathtt{aabaabab}$.
  Recall that the automaton $A$ such that $L(A)=L(w)$ was shown in \cref{fig:automaton}.
  Consider a word $x=\mathtt{aabbabba}$ of the same length as $w$.
  For this word $x \not \in L(A)$ and $x^2 \in L(A)$.
  Black circles below represent the states of the automaton $A$ after processing the subsequent letters of $x^2$:
  
  
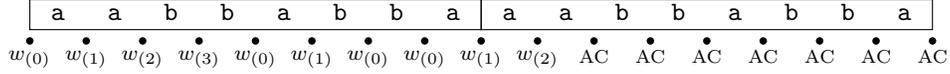
\begin{center}
    \begin{tikzpicture}[scale=1,xscale=1.5]

\foreach \x/\c in {
  0.5/a, 1/a, 1.5/b, 2/b, 2.5/a, 3/b, 3.5/b, 4/a
}{
  \draw (\x,-.4) node[above] {\tt \c};
  \draw[xshift=4cm] (\x,-.4) node[above] {\tt \c};
}

\draw[xshift=0.25cm,yshift=-0.2cm] (0,-.2) rectangle (4,.2);
\draw[xshift=4.25cm,yshift=-0.2cm] (0,-.2) rectangle (4,.2);

\begin{scope}[xshift=0.25cm,yshift=0.2cm]
\foreach \x in {0,1,...,16} {
\node[fill=black, circle, inner sep = 1] (q\x) at (\x/2 , -.75) {};
}

\draw (q0) node[below] {\footnotesize{$w_{(0)}$}};
\draw (q1) node[below] {\footnotesize{$w_{(1)}$}};
\draw (q2) node[below] {\footnotesize{$w_{(2)}$}};
\draw (q3) node[below] {\footnotesize{$w_{(3)}$}};
\draw (q4) node[below] {\footnotesize{$w_{(0)}$}};
\draw (q5) node[below] {\footnotesize{$w_{(1)}$}};
\draw (q6) node[below] {\footnotesize{$w_{(0)}$}};
\draw (q7) node[below] {\footnotesize{$w_{(0)}$}};
\draw (q8) node[below] {\footnotesize{$w_{(1)}$}};
\draw (q9) node[below] {\footnotesize{$w_{(2)}$}};
\draw (q10) node[below] {\scriptsize{AC}};
\draw (q11) node[below] {\scriptsize{AC}};
\draw (q12) node[below] {\scriptsize{AC}};
\draw (q13) node[below] {\scriptsize{AC}};
\draw (q14) node[below] {\scriptsize{AC}};
\draw (q15) node[below] {\scriptsize{AC}};
\draw (q16) node[below] {\scriptsize{AC}};
\end{scope}

\end{tikzpicture}

  \end{center}

  For this word the decomposition of \cref{lem:comb_crucial} is as follows:

  
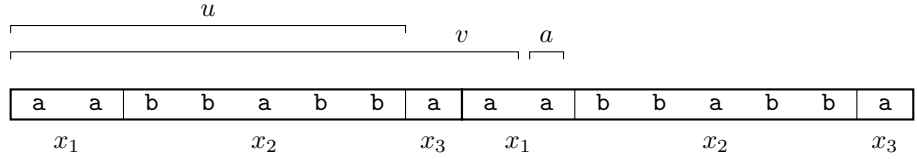
\begin{center}
    \begin{tikzpicture}[scale=1,xscale=1.5]

\foreach \x/\c in {
  0.5/a, 1/a, 1.5/b, 2/b, 2.5/a, 3/b, 3.5/b, 4/a
}{
  \draw (\x,-.4) node[above] {\tt \c};
  \draw[xshift=4cm] (\x,-.4) node[above] {\tt \c};
}

\draw[xshift=0.25cm,yshift=-0.2cm,thick] (0,-.2) rectangle (4,.2);
\draw[xshift=4.25cm,yshift=-0.2cm,thick] (0,-.2) rectangle (4,.2);
\draw[xshift=0.25cm,yshift=-0.2cm] (1,-.2) -- (1,.2)  (3.5,-.2) -- (3.5,.2);
\draw[xshift=4.25cm,yshift=-0.2cm] (1,-.2) -- (1,.2)  (3.5,-.2) -- (3.5,.2);

\begin{scope}[yshift=-0.5cm]
\draw[xshift=0.25cm] (0.5,0) node[below] {$x_1$};
\draw[xshift=0.25cm] (2.25,0) node[below] {$x_2$};
\draw[xshift=0.25cm] (3.75,0) node[below] {$x_3$};
\end{scope}

\begin{scope}[xshift=4cm,yshift=-0.5cm]
\draw[xshift=0.25cm] (0.5,0) node[below] {$x_1$};
\draw[xshift=0.25cm] (2.25,0) node[below] {$x_2$};
\draw[xshift=0.25cm] (3.75,0) node[below] {$x_3$};
\end{scope}

\begin{scope}[xshift=0.25cm]
\draw[yshift=0.15cm] (0,.6) -- (0,.7) -- node[above]{$u$} (3.5,.7) -- (3.5,.6);
\draw[yshift=0.1cm] (0,.3) -- (0,.4) -- (4.5,.4) -- (4.5,.3);
\draw[yshift=0.1cm] (4,.4) node[above]{$v$};
\draw[yshift=0.1cm] (4.6,.3) -- (4.6,.4) -- node[above] {$a$} (4.9,.4) -- (4.9,.3);
\end{scope}

\end{tikzpicture}

  \end{center}

  In this case in the proof of the lemma we have $u=\mathtt{aabbabb}$,
  $v=\mathtt{aabbabbaa}$, and $k=2$.
\end{example}

Denote $\pi_k(i,j) = |L_A(w_{(i)},w_{(j)}) \cap \Sigma^k|$.
We say that a number is {\em small} if it fits into a constant number of machine words or,
in other words, is polynomial with respect to $n+\sigma$.
Using \cref{lem:comb_crucial}, we obtain a formula for $|\CS(w)|$.

\begin{lemma}\label{lem:cs}
  For every self-minimal word $w\in \Sigma^n$, there exist coefficients $\alpha_{i,j}$ that are small numbers such that
  $$|\CS(w)| = \pi_n(0,n)+\sum_{i,j=0}^n \alpha_{i,j} \pi_j(i,0).$$
  Moreover, the coefficients $\alpha_{i,j}$ can all be computed in $\Oh(n^2)$ time.
\end{lemma}
\begin{proof}
  We apply \cref{fct:char} with \cref{cor:auto} and actually compute
  $|\{x\in \Sigma^n: x^2\in L(A)\}|$.
  If $x\in L(A)$, then obviously $x^2\in L(A)$.
  For this part, we need to compute $|L(A) \cap \Sigma^n|$, which is exactly $\pi_n(0,n)$.
  Now it suffices to count $x\in \Sigma^n$ such that $x^2\in L(A)$ but $x\notin L(A)$.

  Let us recall the characterization of such words from \cref{lem:comb_crucial}.
  We consider all $\Oh(n^2)$ choices of $|x_1|$ and $|x_3|$,
  and count the number of $x$'s conditioned on these values.
  Let $x_1 = x'_1a$ where $x_1'\in \Sigma^*$, $a\in \Sigma$.
  Note that $x_3x_1=x_3x'_1a\in \Pref_{-}(w)$, so $x_3x'_1$ is a prefix $w_{(k)}$ of $w$
  and $\delta(w_{(k)},a) = AC$.
  Hence, $k$ is uniquely determined by $|x_1|$ and $|x_3|$.
  In particular, $x_3=w[1,|x_3|]$ and $x'_1=w[|x_3|+1,k]$ are uniquely determined.
  Let us define $\ell$ as $w_{(\ell)}=\delta(w_{(0)},x'_1)$;
  see \cref{fig:proof}.
  To efficiently determine $\ell$ for each choice of $|x_1|$ and $|x_3|$, we precompute $\delta(w_{(0)},w[i,j])$ for all factors $w[i,j]$ of $w$.
  Since \[\delta(w_{(0)},w[i,j+1])=\delta(\delta(w_{(0)},w[i,j]),w[j+1]),\] these values can be computed in $\Oh(n)$ time for each $i$,
  i.e., $\Oh(n^2)$ time in total.

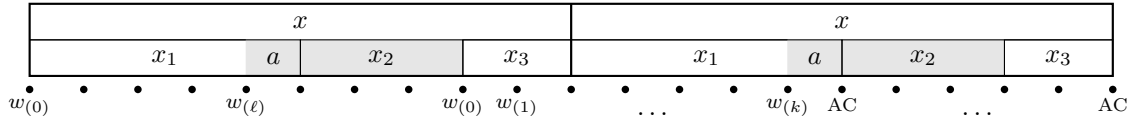
\begin{figure}[htpb]
\centering{
  \begin{tikzpicture}[scale=1.2,xscale=1.2]

\draw (0,-.2) rectangle (5,.2);
\draw (2.5, 0) node{$x$};
\draw (5,-.2) rectangle (10,.2);
\draw (7.5, 0) node{$x$};
\draw (1.25, -.4) node{$x_1$};
\filldraw[black!10] (2,-.2) rectangle (4, -.6);
\draw (0,-.2) rectangle (2.5, -.6);
\draw (2.5,-.2) rectangle (4, -.6);
\draw (2.25, -.4) node{$a$};
\draw (3.25, -.4) node{$x_2$};
\filldraw[black!10] (7,-.2) rectangle (9, -.6);
\draw (4,-.2) rectangle (5, -.6);
\draw (4.5, -.4) node{$x_3$};
\draw (5,-.2) rectangle (7.5, -.6);
\draw (6.25, -.4) node{$x_1$};
\draw (7.5,-.2) rectangle (9, -.6);
\draw (7.25, -.4) node{$a$};
\draw (9,-.2) rectangle (10, -.6);
\draw (8.25, -.4) node{$x_2$};
\draw (9.5, -.4) node{$x_3$};

\foreach \x in {0,1,...,20} {
\node[fill=black, circle, inner sep = 1] (q\x) at (\x/2 , -.75) {};
}

\draw[thick] (0,-.6) rectangle (10,.2);
\draw[thick] (5,-.6) -- (5,.2);

\draw (q20) node[below] {\scriptsize{AC}};
\draw (q15) node[below] {\scriptsize{AC}};
\draw (q14) node[below] {\footnotesize{$w_{(k)}$}};
\draw (q8) node[below] {\footnotesize{$w_{(0)}$}};
\draw (q9) node[below] {\footnotesize{$w_{(1)}$}};
\draw (q4) node[below] {\footnotesize{$w_{(\ell)}$}};
\draw (q0) node[below] {\footnotesize{$w_{(0)}$}};

\draw (17.5/2, -.9) node[below] {$\ldots$};
\draw (11.5/2, -.9) node[below] {$\ldots$};

\end{tikzpicture}
}
\caption{\label{fig:proof}
  Illustration of \cref{lem:cs}.
  Both lines represent different factorizations of the same word $x^2$.
  Black circles represent states of the automaton.
  Only shaded letters are not necessarily uniquely determined by $|x_3|$ and $|x_1|$ for a fixed $w$.
}
\end{figure}
  
  Once we know $\ell$, we need to count
  $$\left\{ax_2 \in \Sigma^{n-k}\;:\; a\in\Sigma \text{, }\delta(w_{(k)},a)=AC\text{, and } ax_2\in L_A(w_{(\ell)},w_{(0)}) \right\}.$$
  Note that $\delta(w_{(\ell)},a)\in\{w_{(0)},w_{(\ell+1)}\}$, since $\delta(w_{(\ell)},a)=AC$ would
  imply that $x \in L(A)$.
  Thus, the number of words $ax_2$ is equal to
  \begin{equation}\label{eq:gamma}
    \sum_{q\in \{0,\ell+1\}}\ \gamma(k,\ell,q)\,\pi_{n-k-1}(q,0), \quad\mbox{where } \gamma(k,\ell,q)=|\{a\in \Sigma: \delta(w_{(\ell)},a)=q\wedge \delta(w_{(k)},a)=AC\}|.
  \end{equation}
  Each coefficient $\gamma(k,\ell,q)$ can be computed in constant time, since in our automaton $A$
  the transition function $\delta$ has an especially simple form.
  By rearranging the summands of \eqref{eq:gamma}, we obtain a formula for $|\CS(w)|$
  in the desired form.
\mayqed\end{proof}

\section{Ranking Lyndon Words with $O(n^2)$ Arithmetic Operations}\label{sec:arith}
In this section by arithmetic operations we mean addition, subtraction and multiplication.
The following lemma shows how to efficiently count certain walks in the automaton $A$ recognizing $L(w)$.
Its proof is based on vector-matrix multiplication.
\begin{lemma}\label{lem:mat}
  Let $A=(Q,\Sigma, \delta, q_0,F)$ be a sparse deterministic automaton with $n$ states.
  Given $q\in Q$ and  $m\in\mathbb{Z}_{\ge 0}$, it takes $\Oh(mn)$ arithmetic operations on integers of magnitude $\sigma^m$ to compute
  all values $|L_A(q,q')\cap \Sigma^k|$ and $|L_A(q',q)\cap \Sigma^k|$ for $0\le k \le m$, $q'\in Q$.
\end{lemma}
\begin{proof}
  We construct an $n\times n$ matrix $M$ with rows and columns indexed by states from $Q$.
  Set $M_{q,q'}=|\{a\in \Sigma : \delta(q,a)=q'\}|$.
  It is easy to see that $(M^k)_{q,q'}=|L_A(q,q')\cap \Sigma^k|$.
  Consequently, the entries of $M^k$ belong to $\{0,\ldots,\sigma^k\}$.
  
  Note that the matrix $M$ is sparse, i.e., it contains $\Oh(n)$ non-zero entries.
  Thus, for a (vertical) vector $\mathbf{v}$ one can compute $M\mathbf{v}$ and $\mathbf{v}^TM$ using $\Oh(n)$ 
  arithmetic operations.
  For $q\in Q$ let $\mathbf{e}_q$ be the unit vector with one at the position corresponding to $q$.
  Observe that $(M^k)_{q,q'}$ is equal to the $q'$-th entry of $\mathbf{e}_q^TM^k$.
  For a fixed state $q\in Q$, we can compute these (horizontal) vectors for $0\le k\le m$ using
  $m$ vector-matrix multiplications.
  Symmetrically, $|L_A(q',q)\cap \Sigma^k|=(M^k)_{q',q}$ is the $q'$-th entry of $M^k\mathbf{e}_{q}$,
  and we can also compute these (vertical) vectors for $0\le k\le m$ using
  $m$ matrix-vector multiplications.
  In total, we perform $\Oh(mn)$ arithmetic operations.
\mayqed\end{proof}

\noindent
The algorithm below combines the results obtained so far to provide the implementation for \cref{lem:csgen}\ref{aaa}.

\begin{proof}[Proof (of \cref{lem:csgen}\ref{aaa})]
Our algorithm is based on the formula of \cref{lem:cs},
whose proof already provides a procedure to compute the coefficients $\alpha_{i,j}$.
On the other hand, \cref{lem:mat} states that values $\pi_j(0,i)$ and $\pi_j(i,0)$ for $0\le i,j \le n$
can be determined using $\Oh(n^2)$ arithmetic operations given the automaton recognizing $L(w)$.

\begin{myalgorithm}[H]{13 cm}
  \Algo{Computing $|\CS(w)|$ in $\Oh(n^2)$ time in the unit-cost RAM model}
  {
    Construct automaton $A$ for $w$ \{ \cref{cor:auto} \}\\
    Compute $\pi_j(0,i)$ and $\pi_j(i,0)$ for all $0 \le i,j \le n$ \{ \cref{lem:mat} \}\\
    Compute $\alpha_{i,j}$ coefficients \{ \cref{lem:cs} \}\\
    $|\CS(w)| := \pi_n(0,n)+\sum_{i,j=0}^n \alpha_{i,j} \pi_j(i,0)$ \{ \cref{lem:cs} \}
  }
\end{myalgorithm}

Thus, our algorithm, given in pseudocode above, performs $\Oh(n^2)$ arithmetic operations on integers of magnitude $\sigma^n$ to compute $|\CS(w)|$
for a self-minimal word $w$.
\mayqed\end{proof}

In the unit-cost RAM model of arithmetic operations, we obtain $\Oh(n^2)$ time.
It is easy to check that all arithmetic operations performed in the algorithm above
are additions and subtractions of numbers not exceeding $\sigma^n$
and multiplications of such numbers by small numbers.
Hence, in the word RAM model we obtain $\Oh(n^3)$ time.
In the following section we give an algorithm working in $\Oh(n^2 \log \sigma)$ time in the word RAM model.

\section{Ranking Lyndon Words in $\Oh(n^2 \log \sigma)$ Time on Word RAM}\label{sec:wordram}
The improvement of the time complexity requires a modification of the formula
of \cref{lem:cs}, after which we perform $\Oh(n^2)$ arithmetic operations only on small integers
and only $\Oh(n)$ operations on large integers.
We also use Newton's iteration for power series inversion (\cite{Sieveking}; see also \cite[p.\ 140]{GeddesCL}):

\begin{fact}\label{fct:reciprocal}
  Let $T(n)$ be the time necessary to compute the inverse of a power series $G(x)$ of degree $n$ modulo $x^n$,
  that is, the time to compute a power series $F(x)$ of degree $n$ such that $F(x)G(x) \equiv 1 \pmod{x^n}$.
  Then $T(n)$ satisfies:
  $$T(2^k) \le T(2^{k-1})+cM(2^{k-1})$$
  where $c>0$ is a constant and $M(n)$ is the time to multiply two polynomials of degree $n$
  with coefficients of magnitude not exceeding the $n$-th coefficient of $F(x)$.
\end{fact}

\noindent
For an efficient implementation of \cref{fct:reciprocal}, we use an integer multiplication algorithm
designed for the word RAM model; see F\"urer~\cite{DBLP:conf/latin/Furer14a}.

\begin{lemma}\label{lem:multiply}
  Two polynomials of degree at most $n$ with coefficients of magnitude $\sigma^n$
  can be multiplied in $\Oh(n^2 \log \sigma)$ time in the word RAM model.
\end{lemma}
\begin{proof}
  Let $F(x)$ and $G(x)$ be the considered polynomials.
  We encode them as integers $u$ and $v$ as follows.
  Both $u$ and $v$ are divided into $n$ chunks consisting of $n \log \sigma + \log n$ bits each.
  The $i$-th least significant chunk of $u$ (respectively $v$) holds the $i$-th coefficient of $F(x)$ (respectively $G(x)$)
  prepended by zeroes.
  Then the corresponding chunks of $uv$ hold the coefficients of $F(x)G(x)$.
  Both numbers $u$ and $v$ have $\Oh(n^2 \log \sigma)$ bits.
  Therefore, the product $uv$ can be computed in $\Oh(n^2 \log \sigma)$ time \cite{DBLP:conf/latin/Furer14a}.
\mayqed\end{proof}

\noindent
With the auxiliary \cref{fct:div}, we obtain the following tool.

\begin{lemma}\label{lem:generating}
  Let $F(x)$ and $G(x)$ be power series such that $F(x) G(x) \equiv 1$.
  Assume that the $k$-th coefficient of $F(x)$ is of magnitude $\sigma^k$.
  If the coefficients of $G(x)$ can be computed in $\Oh(1)$ time,
  then $F(x) \bmod x^n$ can be computed in $\Oh(n^2 \log \sigma)$ time
  in the word RAM model.
\end{lemma}

Now we show how to use Lemma~\ref{lem:generating} to count specific paths in
the automaton $A$ for the word $w$.
Denote
$$T_i = \pi_i(0,0) \quad\mbox{and}\quad a_i = |\{c\in \Sigma : \delta(w_{(i-1)}, c) = w_{(0)}\}|=|\{ c \in \Sigma: c > w[i]\}|.$$

\begin{lemma}\label{lem:T}
  All values $T_0,\ldots,T_n$ can be computed in $\Oh(n^2 \log \sigma)$ time in the word RAM model.
\end{lemma}
\begin{proof}
  Assume that for $k<0$, $T_k=0$.
  Recall that a non-empty path from $w_{(0)}$ to itself in $A$ passes through a number of consecutive states
  $w_{(1)},w_{(2)},\ldots,w_{(i)}$ before it first comes back to $w_{(0)}$.
  Hence, $T_k$ satisfy the following recurrence:
  $$
    T_k = \begin{cases}
      0&\text{for }k<0,\\
      1&\text{for }k=0,\\
      a_1T_{k-1}+\ldots+a_nT_{k-n}&\text{otherwise}.
    \end{cases}
  $$
  Let us set $a_0=-1$.
  Let $F$ and $G$ be the generating functions of $T_k$ and $a_k$:
  $$F(x)=\sum_{k=0}^\infty T_k x^k,\quad G(x)=\sum_{k=0}^n a_kx^k.$$
  Note that:
  \begin{align*}
    F(x)G(x) &= \sum_{k=0}^\infty x^k\sum_{m=0}^k a_m T_{k-m}
    = -1 + \sum_{k=1}^\infty x^k\sum_{m=0}^n a_m T_{k-m} \\
    &= -1 + \sum_{k=1}^\infty x^k (-T_k + \sum_{m=1}^n a_m T_{k-m})
    = -1.
  \end{align*}
  This concludes that we can use \cref{lem:generating} to compute $n$ first coefficients of $F(x)$
  in $\Oh(n^2 \log \sigma)$ time.
\mayqed\end{proof}

We extend the results of the previous lemma to compute the first term of the formula for $|\CS(w)|$.

\begin{lemma}\label{lem:complex}
  The value $\pi_n(0,n)$ can be computed in $\Oh(n^2 \log \sigma)$ time in the word RAM model.
\end{lemma}
\begin{proof}
  Note that
  \begin{equation}\label{eq:complex}
    \pi_n(0,n) = \sum_{i=0}^{n-1} T_i c_{n-i}
  \end{equation}
  where $c_j$ is the number of paths of length $j$ that start in $w_{(0)}$, end in $AC$ and do not pass through $w_{(0)}$ again.
  Denote $a'_i=|\{a \in \Sigma : \delta(w_{(i-1)}, a) = AC\}|$. Note that $a'_i = \sigma-1-a_i$ for $i < n$ and $a'_n = \sigma-a_n$.
  Moreover, for every $j\in\{1,\ldots,n\}$,
  $$c_j=a'_1\sigma^{j-1}+a'_2\sigma^{j-2}+\ldots+a'_j$$
  as in the considered path we traverse some number of edges $k \in \{0,\ldots,j-1\}$ passing through $w_{(0)},\ldots,w_{(k)}$,
  then we use an edge to the accepting state and stay in that state for the remaining $j-1-k$ steps.

  Due to the recurrence $c_{j+1}=\sigma c_j+a'_{j+1}$, all values $c_j$ can be computed in $\Oh(n^2)$ time.
  By \cref{lem:T}, all values $T_j$ can be computed in $\Oh(n^2 \log \sigma)$ time.
  Obviously $c_j, T_j \le \sigma^j$.
  This concludes that we can use the algorithm of \cref{lem:multiply} to multiply
  the polynomials
  $$F(x)=\sum_{i=0}^{n-1} T_i x^i \quad\mbox{and}\quad G(x)=\sum_{i=0}^{n-1} c_{i+1}x^i.$$
  The coefficient of $F(x)G(x)$ at $x^{n-1}$ is exactly the desired sum \eqref{eq:complex}.
\end{proof}

Finally, we are ready to prove \cref{lem:csgen}\ref{bbb}. To this end, we show that the remaining terms of the formula for $|\CS(w)|$ can be computed efficiently in the word RAM model.

\begin{proof}[Proof (of \cref{lem:csgen}\ref{bbb})]
  We provide an efficient implementation of the formula from \cref{lem:cs}.
  For the $\pi_n(0,n)$ part we use \cref{lem:complex}.
  Now we show how to transform the coefficients $\alpha_{i,j}$ to obtain an equivalent set of small coefficients $\beta_{i,j}$ satisfying
  $\beta_{i,j} \ne 0$ if and only if $i=0$ or $j=0$.
  We use the following claim.

  \medskip
  \begin{claim} For $0\le i < n$ and $j \ge 1$, we have
    \begin{equation}\label{rec}
      \pi_j(i,0)=\pi_{j-1}(i+1,0)+a_{i+1}\pi_{j-1}(0,0).
    \end{equation}
    Moreover, $\pi_{j}(n,0)=0$ for $j \ge 0$.
  \end{claim}

  \medskip
  The formula \eqref{rec} corresponds to traversing the first edge of the path from $i$ to $0$.
  We arrive at the following algorithm which reduces computation of the required sum of a quadratic number
  of large numbers to the computation of a linear combination of only
  linearly many big numbers $T_j$.

  \begin{myalgorithm}[H]{9 cm}
    \Algo{Compute $|\CS(w)|$}
    {
    \ForEach{$i,j\in \{0,\ldots,n\}$}{
      $\beta_{i,j}:=\alpha_{i,j}$
    }
    \For{$j:=n$ {\rm\textbf{downto}} $1$}{
      \For{$i:=1$ {\rm\textbf{to}} $n-1$}{
        $\beta_{i+1,j-1}\,\,\mbox{+=}\,\,\beta_{i,j}$\\
        $\beta_{0,j-1}\,\,\mbox{+=}\,\,a_{i+1}\beta_{i,j}$\\
        $\beta_{i,j}:=0$
      }
    }
    \Return{\hspace*{0.2cm}$\pi_n(0,n)\;+\;\sum_{j=0}^n\beta_{0,j}\cdot T_j$}
    }
  \end{myalgorithm}

  Denote $A\,=\,\sum_{i,j=0}^n \beta_{i,j} \pi_j(i,0)$.  By \eqref{rec} we have:
  $$A\;=\;A-\beta_{i,j}\pi_j(i,0)\,+\, \beta_{i,j}\pi_{j-1}(i+1,0)+\beta_{i,j} a_{i+1} \pi_{j-1}(0,0).$$
  Consequently, resetting $\beta_{i,j}$ to zero and increasing the coefficients $\beta_{i+1,j-1}$ and $\beta_{0,j-1}$
  in the inner iteration does not alter the total sum $A$. 
  Hence, after every iteration of the inner for-loop the coefficients satisfy the following invariant:
  $$A\,=\, \sum_{i,j=0}^n \beta_{i,j} \pi_j(i,0)=\sum_{i,j=0}^n \alpha_{i,j} \pi_j(i,0).$$

  Observe that once $\beta_{i,j}$ is reset to zero, it will not be changed anymore.
  Hence, at the end of the algorithm we have $\beta_{i,j}=0$ if $j>0$ and $1\le i \le n-1$. Note that $\beta_{n,j}\cdot \pi_j(n,0)=\beta_{n,j}\cdot 0 = 0$
  for each $j$ and $\beta_{i,0}\cdot \pi_0(i,0) = \beta_{i,0} \cdot 0 = 0$ for $i\ne 0$.
  This concludes that at the end of the algorithm we have
  $$ \sum_{i,j=0}^n \alpha_{i,j} \pi_j(i,0)\;=\;\sum_{i,j=0}^n \beta_{i,j} \pi_j(i,0)\;=\;\sum_{j=0}^n\beta_{0,j}\cdot T_j.$$

  Note that each $\alpha_{i,j}$ coefficient accounts in $\sum_j \beta_{0,j}$
  as at most $(a_{i+1}+a_{i+2}+\ldots+a_{n})\alpha_{i,j}$.
  Hence, the sum of the resulting non-zero coefficients $\beta_{i,j}$
  does not exceed $\sigma n$ times the sum of the initial values $\alpha_{i,j}$.
  At the end, we are to compute a linear combination of $T_j$ with small coefficients.
  Consequently, \cref{lem:T} yields an $\Oh(n^2 \log \sigma)$-time algorithm on the word RAM.
\mayqed\end{proof}

\section{Decoding Minimal de Bruijn Sequence}\label{sec:debruijn}
In this section we focus on decoding lexicographically minimal de Bruijn sequence $\dB_n$ over $\Sigma$:
we aim at an efficient algorithm that for every $w \in \Sigma^n$ computes $\occ(w,\dB_n)$, that is,
the position of the sole occurrence of $w$ in $\dB_n$.
Recall that by $\L^{(n)}$ we denote the set of Lyndon words over $\Sigma$ whose length is a divisor of $n$.
A theorem of Fredricksen and Maiorana \cite{DBLP:journals/jct/Fredricksen70,fredricksen1978necklaces,Knuth}
states that $\dB_n$ is a concatenation of the Lyndon words from $\L^{(n)}$
in the lexicographic order.
The proof of the theorem is constructive,
i.e., for any word $w$ of length $n$ it shows the concatenation of
a constant number of consecutive Lyndon words from the cyclic version of the sequence $\L^{(n)}$ that contain $w$.
This, together with the following lemma which relates $\dB_n$ to $\CS$, lets us compute the
exact position where $w$ occurs in $\dB_n$.

\begin{lemma}\label{lem:db}
  Let $w\in \Sigma^n$ and $\L(w)=\{\lambda\in \L^{(n)}: \lambda^{n/|\lambda|} \le w\}$.
  Then the concatenation, in lexicographic order, of words $\lambda \in \L(w)$
  forms a prefix of $\dB_n$ and its length, $\sum_{\lambda\in \L(w)}|\lambda|$, is equal
  to $|\CS(w)|$.
\end{lemma}
\begin{proof}
  First note that, by \cref{fct:lyndon}\ref{bb}, the lexicographic order of elements $\lambda \in \L^{(n)}$
  coincides with the lexicographic order of $\lambda^{n/|\lambda|}$.
  This shows that the concatenation of elements of $\L(w)$ indeed forms a prefix of $\dB_n$.
  
  It remains to show that $\sum_{\lambda\in \L(w)}|\lambda| = |\CS(w)|$.
  For this we shall build a mapping $\phi : \Sigma^n\to \L^{(n)}$
  such that $|\phi^{-1}(\lambda)|=|\lambda|$ and $\minrot{x}\le w$ for $x \in \Sigma^n$ if and only if $\phi(x)\in \L(w)$.

  Let $x\in \Sigma^n$. There is a unique primitive word $y$ and a positive integer $k$ such that $x=y^k$.
  We set $\phi(x)=\minrot{y}$.
  Note that $\phi(x)$ indeed belongs to $\L^{(n)}$.
  Moreover, to each Lyndon word
  $\lambda$ of length $d\mid n$ we have assigned $v^{n/d}$ for each cyclic rotation $v$ of $\lambda$.
  Thus $|\phi^{-1}(\lambda)| = |\lambda|$.
  Also, $\minrot{x}=\minrot{y}^{n/d}$, so $\minrot{x}\le w$ if and only if $\phi(x)^{n/d}\le w$,
  i.e., $\phi(x)\in \L(w)$.
\mayqed\end{proof}

\begin{theorem}\label{thm:decoding}
  Given a word $w \in \Sigma^n$, the position $\occ(w,\dB_n)$ can be found in $\Oh(n^2\log\sigma)$ time in the word RAM model or $\Oh(n^2)$ time in the unit-cost RAM model.
\end{theorem}
\begin{proof}
  Let $\lambda_1 < \lambda_2 < \cdots < \lambda_p$ be all Lyndon words in $\L^{(n)}$
  (we have $\lambda_1 \lambda_2 \cdots \lambda_p = \dB_n$).
  The proof of the theorem of Fredricksen and Maiorana \cite{fredricksen1978necklaces,Knuth}
  describes the occurrence of $w$ in $\dB_n$, which can be stated succinctly as follows.
  \begin{claim}[Fredricksen and Maiorana~\cite{fredricksen1978necklaces}, Knuth~\cite{Knuth}]
       Assume that $w=(\alpha\beta)^d$, where $d \in \integ_+$ and $\beta\alpha = \lambda_k \in \L^{(n)}$.
     Denote $a = \min\Sigma$ and $z = \max\Sigma$.  
  \begin{enumerate}[label={(\alph*)}] 
    \item\label{it:simple} If $w=z^ia^{n-i}$ for $1 \le i\le n$, then $w$ occurs in $\dB_n$ at position $\sigma^n-i+1$.
    \item\label{it:b} If $\alpha \ne z^{|\alpha|}$, then $w$ is a factor of $\lambda_k\lambda_{k+1}$.
    \item\label{it:c} If $\alpha = z^{|\alpha|}$ and $d>1$, then $w$ is a factor of $\lambda_{k-1}\lambda_k\lambda_{k+1}$.
    \item\label{it:d} If $\alpha = z^{|\alpha|}$ and $d=1$, then $w$ is a factor of $\lambda_{k'-1}\lambda_{k'}\lambda_{k'+1}$,
    where $\lambda_{k'}$ is the largest $\lambda \in \L^{(n)}$ such that $\lambda < \beta$.
  \end{enumerate}
  \end{claim}
  In case \ref{it:simple}, it is easy to locate $w$ in $\dB_n$.
  Further on, we consider only the cases \ref{it:b}, \ref{it:c}, and \ref{it:d}.

  Note that $\lambda_k$ can be retrieved as the primitive root of $\minrot{w}$. For $\lambda_{k'}$ we use the following characterization.
    \begin{claim}
      The Lyndon word $\lambda_{k'}$ is the lexicographically larger among the following two strings:
    \begin{enumerate}[label={(\arabic*)}]
    \item\label{ONE} the longest proper prefix of $\beta$ contained in $\L^{(n)}$;
    \item\label{TWO} the primitive root of the largest self-minimal word $w'\in \Sigma^n$ such that $w'<\beta$.
    \end{enumerate}
    \end{claim}
    \begin{proof}
  The definition of $\lambda_{k'}$ yields $\lambda_{k'} < \beta \le \lambda_{k'+1}$, whereas
    \cref{fct:lyndon}\ref{aa} implies $\lambda_{k'} \le (\lambda_{k'})^{n/|\lambda_{k'}|} < \lambda_{k'+1}$.
    This gives rise to two cases:
    If $\lambda_{k'} < \beta < (\lambda_{k'})^{n/|\lambda_{k'}|}$, then $\lambda_{k'} \in \L^{(n)}$ is a proper prefix of $\beta$.
    In this case, $\lambda_{k'}$ must also be the lexicographically largest, i.e., longest, prefix of $\beta$ that belongs to $\L^{(n)}$.
    This results in the string from case~\ref{ONE}.
    Otherwise, $(\lambda_{k'})^{n/|\lambda_{k'}|} \le \beta \le \lambda_{k'+1}$.
    By \cref{fct:lyndon}\ref{aa}, $(\lambda_{k'})^{n/|\lambda_{k'}|}$ is the largest self-minimal length-$n$ word that is smaller than $\beta$.
    That is, $(\lambda_{k'})^{n/|\lambda_{k'}|}$ corresponds to $w'$ from case \ref{TWO} and $\lambda_{k'}$ is the primitive root of~$w'$.
    \end{proof}

  The string $\lambda_{k'}$ can be computed efficiently using the above claim.
  In case \ref{ONE}, for each proper prefix of~$\beta$, we can use \cref{fct:minrot} to check in $\Oh(n)$ time if it is a Lyndon word; we then select the longest such prefix.
  In case \ref{TWO}, $w'$ can be computed in $\Oh(n^2)$ time using \cref{lem:prv}; as noted in \cref{sec:prelim}, the primitive root of $w'$ can be computed in $\Oh(n)$ time.
  Finally, selecting the larger of the two candidates takes $\Oh(n)$ time.
  Overall, $\lambda_{k'}$ is computed in $\Oh(n^2)$ time.

  Once we know $\lambda_{k'}$ and $\lambda_{k}$, depending on the case, we need to find the successor in $\L^{(n)}$
  and possibly the predecessor in $\L^{(n)}$ of one of them. 
  For any $\lambda\in \L^{(n)}$, the successor  in $\L^{(n)}$ can be generated by iterating a single step of the FKM algorithm at most
  $(n-1)/2$ times \cite{FK}, i.e., in $\Oh(n^2)$ time.
  For the predecessor in $\L^{(n)}$, a version of the FKM algorithm that visits the Lyndon words
  in reverse lexicographic order can be used \cite{Knuth}.
  It also takes $\Oh(n^2)$ time to find the predecessor.
  In all cases, we obtain in $\Oh(n^2)$ time the Lyndon words whose concatenation contains $w$.
  
  Then, we use exact pattern matching to locate $w$ in the concatenation.
  This gives us the relative position of $w$ in $\dB_n$ with respect to the position of the canonical occurrence of $\lambda_k$ or $\lambda_{k'}$ in $\dB_n$.
  \Cref{lem:db} proves that such an occurrence of $\lambda\in \L^{(n)}$ \emph{ends} at position
  $|\CS(\lambda^{\frac{n}{|\lambda|}})|$, which can be computed in $\Oh(n^2\log\sigma)$ time in the word RAM model
  or $\Oh(n^2)$ time in the unit-cost RAM model by \cref{lem:csgen}.
  Applied to $\lambda_k$ or $\lambda_{k'}$, this concludes the proof.
\mayqed\end{proof}

\begin{example}
  Below we present the four cases of the claim in the proof of \cref{thm:decoding}
  on the sequence $\dB_6$ over a binary alphabet (i.e., the lexicographically minimal binary de Bruijn sequence of rank 6), which has the following
  decomposition into Lyndon words $\lambda_1,\lambda_2,\ldots, \lambda_{14}$:
  \begin{center}
\begin{tikzpicture}[scale=0.2]

\foreach \x/\c in {
1/0,
3/0, 4/0, 5/0, 6/0, 7/0, 8/1,
10/0, 11/0, 12/0, 13/0, 14/1, 15/1,
17/0, 18/0, 19/0, 20/1, 21/0, 22/1,
24/0, 25/0, 26/0, 27/1, 28/1, 29/1,
31/0, 32/0, 33/1,
35/0, 36/0, 37/1, 38/0, 39/1, 40/1,
42/0, 43/0, 44/1, 45/1, 46/0, 47/1,
49/0, 50/0, 51/1, 52/1, 53/1, 54/1,
56/0, 57/1,
59/0, 60/1, 61/0, 62/1, 63/1, 64/1,
66/0, 67/1, 68/1,
70/0, 71/1, 72/1, 73/1, 74/1, 75/1,
77/1
}{
  \draw (\x,-.4) node[above] {\c};
}

\draw (77.5,1.8) -- (73.5,1.8) -- (73.5,-0.3) -- (77.5,-0.3);
\draw (0.5,1.8) -- (4.5,1.8) -- (4.5,-0.3) -- (0.5,-0.3);
\draw (73.5,-1) node[below right] {\small\ref{it:simple} 111000};

\draw (11.5,-0.3) rectangle (18.5,1.8);
\draw (15,-1) node[below] {\small\ref{it:b} 001100};

\draw (62.5,-0.3) rectangle (70.5,1.8);
\draw (66.5,-1) node[below] {\small\ref{it:c} 110110};

\draw (27.5,-0.3) rectangle (35.5,1.8);
\draw (31.5,-1) node[below] {\small\ref{it:d} 110010};

\foreach \x/\c in {
  1/1, 5.5/2, 12.5/3, 19.5/4, 26.5/5, 32/6, 37.5/7,
  44.5/8, 51.5/9, 56.5/10, 61.5/11, 67/12, 72.5/13, 77/14
}{
  \draw[yshift=0.5cm] (\x,2) node[above] {$\lambda_{\c}$};
}
\end{tikzpicture}
  \end{center}
  \begin{description}
    \item{Case \ref{it:simple}:}
      $\occ(111000,\dB_6)=62$, and $111000$ appears as a factor of $\lambda_{13}\lambda_{14}\lambda_1\lambda_2$.
    \item{Case \ref{it:b}:}
      $\occ(001100,\dB_6)=10$, and $001100$ appears as a factor of $\lambda_3\lambda_4$.
    \item{Case \ref{it:c}:}
      $\occ(110110,\dB_6)=53$, and $110110$ appears as a factor of $\lambda_{11}\lambda_{12}\lambda_{13}$.
    \item{Case \ref{it:d}:}
      $\occ(110010,\dB_6)=24$, and $110010$ appears as a factor of $\lambda_5\lambda_6\lambda_7$.
  \end{description}
\end{example}

\medskip
\noindent
To compute the $k$-th symbol of $\dB_n$, we have to locate the Lyndon word from $\L^{(n)}$ containing
the $k$-th position of $\dB_n$.
We apply binary search as in \cref{thm:kth}.

\begin{theorem}\label{thm:random-access}
  Given integers $n$ and $k$, the $k$-th symbol of $\dB_n$ can be computed
  in $\Oh(n^3\log^2\sigma)$ time in the word RAM model or $\Oh(n^3 \log\sigma)$ time in the unit-cost RAM model.
\end{theorem}
\begin{proof}
  We binary search for the smallest word $v\in \Sigma^n$ 
  such that $|\CS(v)|\ge k$, using \cref{lem:csgen} to test the condition.
  In each step of the binary search, we actually consider a self-minimal word, due to \cref{lem:prv}.
  Therefore the resulting word $v$ is of the form $\lambda^d$ for some $\lambda\in \L^{(n)}$.
  By \cref{lem:db}, a prefix of $\dB_n$ of length $|\CS(v)|$ contains all Lyndon words from $\L(v)$.
  Moreover, by \cref{fct:lyndon}\ref{aa}, this prefix ends with $\lambda$.
  This means that the $k$-th position of $\dB_n$ lies within the canonical
  occurrence of $\lambda$. More precisely, it suffices
  to return the $(|\CS(v)|-k+1)$-th \emph{last} symbol of $\lambda$
  (which is also the $(|\CS(v)|-k+1)$-th last symbol of $v$).
  As in \cref{thm:kth}, the binary search introduces a multiplicative
  $\Oh(n\log\sigma)$ factor to the complexity of the algorithm of \cref{lem:csgen}. 
\mayqed\end{proof}

Recently, Au \cite{DBLP:journals/dm/Au15} introduced a variant of a de Bruijn sequence
in which each (cyclic) factor of length $n$ is primitive and each primitive word from $\Sigma^n$ occurs as a (cyclic) factor.
He also proved that the lexicographically minimal sequence satisfying this condition,
denoted $\dB'_n$, is the concatenation in lexicographic order of Lyndon words of length $n$ over $\Sigma$.

\begin{example}
  For $n=6$ and binary alphabet we have the following decomposition of $\dB'_6$:
{ $$  000001\; 000011\; 000101\; 000111\;  001011\;
  001101\; 001111\;  010111\; 011111.$$
}
\end{example}

The ranking algorithm for Lyndon words lets us derive a counterpart of \cref{thm:decoding}
for $\dB'_n$ with a slightly simpler proof (admitting a similar structure, though).
\begin{proposition}\label{prop:decoding2}
Given a primitive word $w \in \Sigma^n$, $\occ(w,\dB'_n)$
  can be found in $\Oh(n^2\log\sigma)$ time in the word RAM model or $\Oh(n^2)$ time in the unit-cost RAM model.
\end{proposition}
\begin{proof}
  Let $\lambda_1 < \lambda_2 < \cdots < \lambda_p$ be all Lyndon words in $\L_n$
  (we have $\lambda_1 \lambda_2 \cdots \lambda_p = \dB'_n$).
  The proof of a theorem of Au \cite[Theorem 9]{DBLP:journals/dm/Au15}
  describes the occurrence of $w$ in $\dB'_n$, which can be stated succinctly as follows.
  \begin{claim}[Au \cite{DBLP:journals/dm/Au15}]
       Assume that $w=\alpha\beta$ where $\alpha\ne \varepsilon$ and $\beta\alpha=\lambda_k$ is a Lyndon word of length $n$.
     Denote $a = \min\Sigma$ and $z = \max\Sigma$.
  \begin{enumerate}[label={(\alph*)}]
    \item\label{it:simple2} If $w=z^ia^{n-i}$ for $i\ge 1$, then $w$ occurs in $\dB'_n$ at position $|\dB'_n|-i+1$.
    \item\label{it:b2} If $\alpha \ne z^{|\alpha|}$, then $w$ occurs in $\lambda_k\lambda_{k+1}$ at position $1+|\beta|$.
    \item\label{it:d2} If $\alpha = z^{|\alpha|}$, then $w$ occurs in $\lambda_{k'}\lambda_{k'+1}$ at position $1+|\beta|$,
    where $\lambda_{k'}$ is the largest Lyndon word $\lambda\in \L_n$ such that $\lambda < \beta$.
  \end{enumerate}
  \end{claim}

  In case \ref{it:simple2}, it is easy to locate $w$ in $\dB'_n$ with $|\dB'_n| = \sum_{d \mid n} \mu(\tfrac{n}{d})\sigma^d$.
  Otherwise, we observe that $\lambda_k = \minrot{w}$ and this word can be computed using \cref{fct:minrot} along with the decomposition
  $w=\alpha\beta$.
  In case \ref{it:b2}, we observe that the position of $\lambda_k$ in $\dB'_n$
  is $1+n(k-1)$, so $w$ occurs in $\dB'_n$ at position $1+n(k-1)+|\beta|=1+nk-|\alpha|$.
  Thus, it suffices to determine $k=\Lynd(\lambda_k)$ using \cref{thm:lynd}.
  The situation in case \ref{it:d2} is similar: $w$ occurs in $\dB'_n$ at position $1+nk'-|\alpha|$.
  Since $\lambda_{k'}$ is the largest Lyndon word smaller than $\beta$, we have $k'=\Lynd(\beta a^{|\alpha|})$,
  i.e., the computation is also reduced to \cref{thm:lynd}.
\end{proof}
\begin{example}
  Below we present the three cases of the claim in the proof of \cref{prop:decoding2}
  on a sequence $\dB'_6$ over a binary alphabet, which has the following
  decomposition into Lyndon words $\lambda_1,\lambda_2,\ldots, \lambda_{9}$:

  \medskip
  \begin{center}
\begin{tikzpicture}[scale=0.2]

\foreach \x/\c in {
3/0, 4/0, 5/0, 6/0, 7/0, 8/1,
10/0, 11/0, 12/0, 13/0, 14/1, 15/1,
17/0, 18/0, 19/0, 20/1, 21/0, 22/1,
24/0, 25/0, 26/0, 27/1, 28/1, 29/1,
31/0, 32/0, 33/1, 34/0, 35/1, 36/1,
38/0, 39/0, 40/1, 41/1, 42/0, 43/1,
45/0, 46/0, 47/1, 48/1, 49/1, 50/1,
52/0, 53/1, 54/0, 55/1, 56/1, 57/1,
59/0, 60/1, 61/1, 62/1, 63/1, 64/1,
}{
  \draw (\x,-.4) node[above] {\c};
}

\draw (64.5,1.8) -- (61.5,1.8) -- (61.5,-0.3) -- (64.5,-0.3);
\draw (2.5,1.8) -- (5.5,1.8) -- (5.5,-0.3) -- (2.5,-0.3);
\draw (63,-1) node[below] {\scriptsize\ref{it:simple2} 111000};

\draw (11.5,-0.3) rectangle (18.5,1.8);
\draw (13.5,-1) node[below] {\scriptsize\ref{it:b2} 001100};

\draw (27.5,-0.3) rectangle (34.5,1.8);
\draw (30,-1) node[below] {\scriptsize\ref{it:d2} 110010};

\foreach \x/\c in {
  5.5/1, 12.5/2, 19.5/3, 26.5/4, 33.5/5, 40.5/6,
  47.5/7, 54.5/8,61.5/9
}{
  \draw[yshift=0.5cm] (\x,2) node[above] {$\lambda_{\c}$};
}
\end{tikzpicture}
  \end{center}
  \begin{description}
    \item{Case \ref{it:simple2}:}
      $\occ(111000,\dB'_6)=52$, and $111000$ appears as a factor of $\lambda_{9}\lambda_1$.
    \item{Case \ref{it:b2}:}
      $\occ(001100,\dB'_6)=9$, and $001100$ appears as a factor of $\lambda_2\lambda_3$.
    \item{Case \ref{it:d2}:}
      $\occ(110010,\dB'_6)=23$, and $110010$ appears as a factor of $\lambda_4\lambda_5$.
  \end{description}
\end{example}

The $k$-th symbol of $\dB'_n$ is much easier to find than the $k$-th symbol of $\dB_n$,
as shown in the following result.
\begin{proposition}\label{prop:random-access2}
  Given integers $n$ and $k$, the $k$-th symbol of $\dB'_n$ can be computed
  in $\Oh(n^3\log^2\sigma)$ time in the word RAM model or $\Oh(n^3 \log\sigma)$ time in the unit-cost RAM model.
\end{proposition}
\begin{proof}
  The $k$-th symbol of the sequence $\dB'_n$ is the
  $i$-th symbol of the $j$-th Lyndon word of length $n$, where
  $$i=((k-1)\bmod n)+1 \quad\mbox{and}\quad j=\floor{\tfrac{k-1}{n}}+1.$$
  This word can be determined using \cref{thm:kth}.
\end{proof}

\section{Conclusions}\label{sec:concl}

The main result of this paper is an $\Oh(n^2 \log \sigma)$-time algorithm in the word RAM model
and an $\Oh(n^2)$-time algorithm in the unit-cost RAM model for ranking Lyndon words.
We have also presented efficient algorithms for computing a Lyndon word of a given length and rank
in the lexicographic order, decoding lexicographically minimal de Bruijn sequence
of a given rank and computing a particular symbol of this sequence.
Our results can also be applied to ranking necklaces due to a known connection
between Lyndon words and necklaces; see \cite{DBLP:conf/icalp/KoppartyKS14}.

\section{Acknowledgements}

We would like to thank Joe Sawada for making us aware of the work of Kopparty et al.\ \cite{DBLP:conf/icalp/KoppartyKS14}.
We would also like to thank Nguyen Tien Long for spotting a minor error in the proof of \cref{thm:decoding} in a previous version of this manuscript.

\bibliographystyle{plainurl}
\bibliography{lyndon}

\end{document}